\documentclass[a4paper,11pt]{article}
\usepackage{amsmath,amsthm,epic,eepicemu}
\usepackage[pdfpagemode=UseOutlines,
 bookmarks=true,bookmarksopen=true,bookmarksnumbered=true,
 pdffitwindow=true,pdfpagelayout=SinglePage,pdfhighlight=/P,
 colorlinks=true,linkcolor=blue,citecolor=blue,urlcolor=blue]{hyperref}
\usepackage[width=0.9\textwidth,indention=0mm,justification=centerlast,
 font=small,labelfont=bf,labelsep=period]{caption}
\theoremstyle{plain}
\newtheorem{theorem}{Theorem}
\newtheorem{statement}[theorem]{Statement}
\theoremstyle{definition}
\newtheorem{definition}{Definition}
\theoremstyle{remark}
\newtheorem{example}{Example}

\title{Towards the theory of integrable hyperbolic equations of third order}
\author{V.E. Adler, A.B. Shabat}
\date{19 June 2012}

\begin{document}\maketitle

\noindent\hrulefill
\par\noindent {\bf Abstract.}\quad
The examples are considered of integrable hyperbolic equations of third order
with two independent variables. In particular, an equation is found which
admits as evolutionary symmetries the Krichever--Novikov equation and the
modified Landau--Lifshitz system. The problem of choice of dynamical variables
for the hyperbolic equations is discussed.
\medskip

\noindent Keywords: hyperbolic pair, evolutionary symmetry, integrable
hierarchy, B\"acklund variables, Korteweg--de Vries equation, Kaup equation,
Krichever--Novikov equation
\medskip

\noindent MSC: 35L75, 35Q53, 37K10, 37K35

\par\noindent\hrulefill
\smallskip

\section{Introduction}

An important class of nonlinear integrable equations consists of the hyperbolic
ones
\[
 u_{xy}=h(x,y,u,u_x,u_y).
\]
Historically, this is the type of equations which contains the very first
integrable examples, the Liouville and the sine-Gordon equations. The modern
concept of integrability based on the notion of the Lax pair arose first in the
study of evolutionary equations of the KdV type, but its applicability to the
hyperbolic equations was established very soon. Indeed, both classes of
equations are in close relation and existence of a hierarchy of evolutionary
symmetries serves as the most convenient test (or a definition) of
integrability of hyperbolic equations. In particular, some important
classification results were obtained within the symmetry approach
\cite{Zhiber_Shabat_1979, Zhiber_Ibragimov_Shabat_1979, Zhiber_Shabat_1984,
Zhiber_1995, Sokolov_Zhiber_1995, Zhiber_Sokolov_2001}, although the problem of
description of the integrable case is not completely solved so far (see its
review in \cite{Mikhailov_Sokolov_2009}).

The development of the theory shows, on the other hand, that in some cases the
class of equations under consideration should be extended at least to the third
order hyperbolic equations
\[
 u_{xxy}=f(x,y,u,u_x,u_y,u_{xy},u_{xx}).
\]
For instance, the hyperbolic symmetry for the KdV equation itself is of this
form. This is the class of equations which we consider in this paper. Since it
is not very well studied, hence we are not aimed to obtain any classification
result or to derive the necessary integrability conditions. We restrict ourself
by consideration of several interesting examples and discuss the problem of
choice of dynamical variables for the equation.

Let us explain briefly the content of the article. Section \ref{s:hyps}
contains the main definitions, in particular, the notion of consistent pair of
third order hyperbolic equations is introduced. These systems belong to an
intermediate class between the second and the third order equations. Its
consideration is necessary, since the presented examples demonstrate that
systems of this type appears from third order equations under parametric
degeneration.

Sections \ref{s:pot-KdV}, \ref{s:pot-Kaup} are devoted to the examples related
to the KdV and the Kaup equations. These examples are not new, since the
respective hyperbolic equations are equivalent to the Camassa--Holm and the
Degasperis--Procesi \cite{Fuchssteiner_Fokas_1981, Schiff_1998, Hone_1999, 
Ivanov_2006, Degasperis_Procesi, Degasperis_Hone_Holm} equations, up to the 
introducing potential and hodograph type transformations. However, our treatment 
contains some new features since the $y$-symmetries are considered as well. 
We also hope that it is of some methodological value providing an uniform 
approach to these examples.

Section \ref{s:B} is a continuation of section \ref{s:hyps}. Here, we analyze
the consistency condition for a pair of third order hyperbolic equations and
introduce the notion of the B\"acklund variables. This provides an alternative
and more convenient set of dynamical variables, not only for the consistent
pair, but also for a single third order equation.

The main example is considered in section \ref{s:KN}, completely in the
B\"acklund variables. It is related to the Krichever--Novikov equation and
seems to be new, since it hardly could be obtained by use of the standard set
of dynamical variables.

\section{Types of hyperbolic equations}\label{s:hyps}

As it was already said in Introduction, the second order hyperbolic equations
in the light-cone variables
\begin{equation}\label{uxy}
 u_{xy}=h(x,y,u,u_x,u_y)
\end{equation}
belong to the simplest and most well studied class of hyperbolic equations. The
notion of higher order hyperbolic equations can be introduced in many ways
which we do not discuss here, see e.g. \cite{Courant_Hilbert}. The main object
in this paper is a particular class of third order hyperbolic equations with
multiple characteristics, namely of the form
\begin{equation}\label{uxxy}
 u_{xxy}=f(x,y,u,u_x,u_y,u_{xy},u_{xx}).
\end{equation}
For short, we refer to the above types of equations just as to the second and
the third order hyperbolic equations. Another class of equations studied in the
paper consists of the systems of the following type.

\begin{definition}\label{def:pair}
The pair of third order hyperbolic equations
\begin{equation}\label{pair}
\begin{aligned}[b]
 u_{xxy}&=f(x,y,u,u_x,u_y,u_{xy},u_{xx}),\\
 u_{xyy}&=g(x,y,u,u_x,u_y,u_{xy},u_{yy})
\end{aligned}
\end{equation}
is called consistent if the identity holds
\begin{equation}\label{consistency}
 (D_y(f)-D_x(g))\Big|_{\begin{subarray}{c}
  u_{xxy}=f\\ u_{xyy}=g\end{subarray}}=0.
\end{equation}
\end{definition}

\begin{definition}
A consistent pair (\ref{pair}) is called reducible if its general solution
solves some one-parametric family of hyperbolic equations
\[
 u_{xy}=h(\alpha;x,y,u,u_x,u_y),
\]
otherwise the pair is called irreducible.
\end{definition}

\begin{figure}[t]
\def\coord#1{#1\path(35,0)(0,0)(0,35)
\multiput(0,0)(10,0){4}{\circle*{2.5}}
\multiput(0,10)(0,10){3}{\circle*{2.5}}
\matrixput(20,10)(0,10){3}(10,0){2}{\circle{2.5}}
\put(-5,-5){$u$}
\multiputlist(10,-3)(10,0)[ct]{$u_x$,$u_{xx}$}
\multiputlist(-3,10)(0,10)[r]{$u_y$,$u_{yy}$}}
\def\sq{\rule{1.8em}{1.8em}}
\begin{center}
\setlength{\unitlength}{0.18em}
\begin{picture}(160,70)(0,0)
\put(0,10){\coord{\put(0,0){\color[gray]{0.75}\sq}
\multiput(10,10)(0,10){3}{\circle{2.5}}
\put(10,13){$u_{xy}$}}}
\put(60,10){\coord{\put(0,0){\color[gray]{0.6}\sq}
{\color[gray]{0.75}\put(10,0){\sq}\put(0,10){\sq}}
\put(10,10){\circle*{2.5}}\multiput(10,20)(0,10){2}{\circle{2.5}}}
\put(8,23){$u_{xyy}$}\put(18,13){$u_{xxy}$}}
\put(120,10){\coord{\put(0,0){\color[gray]{0.75}\rule{3.6em}{1.8em}}
\multiput(10,10)(0,10){3}{\circle*{2.5}}}
\put(18,13){$u_{xxy}$}}
\end{picture}
\caption{Standard sets of dynamical variables for a hyperbolic second
order equation, for a consistent pair and for a third order equation}
\label{fig:dynvar}
\end{center}
\end{figure}
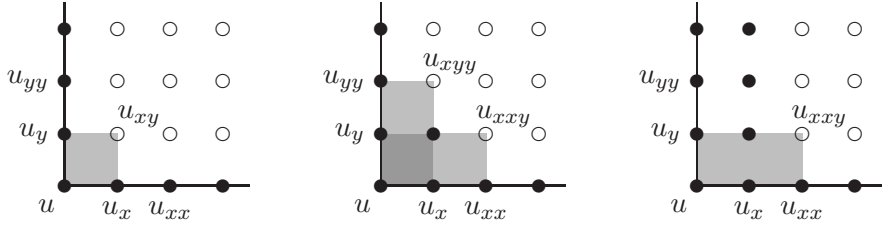

Consistent systems of type (\ref{pair}) are rather delicate generalization of
second order equations. This is clear from comparing the initial data for the
Goursat problem. The role of the Goursat data for equation (\ref{uxy}) can be
played by a pair of functions $u(x,0)=a(x)$, $u(0,y)=b(y)$ such that the
consistency condition $a(0)=b(0)$ is fulfilled. In the case of system
(\ref{pair}) just one additional value should be given, the mixed derivative in
the origin: $u_{xy}(0,0)=\mathop{\rm const}$, while in the case of equation
(\ref{uxxy}) an additional function $u_x(0,y)=c(y)$ is required. The black
disks on fig. \ref{fig:dynvar} mark the dynamical variables for different types
of equations under consideration, that is the set of derivatives in a point
$\partial^m_x\partial^n_y(u)$ which can be chosen independently. We will call
such sets the standard dynamical variables.

\begin{definition}\label{def:integrability}
An equation of any type (\ref{uxy}) or (\ref{uxxy}) or (\ref{pair}) is called
integrable if it is compatible with an infinite hierarchy of evolutionary
symmetries, that is equations of the form
\[
 u_t=f(x,y,[u]),
\]
where the right hand side depends on an arbitrarily large finite number of
dynamical variables for the hyperbolic equation under consideration.
\end{definition}

This definition of integrability is standard enough. More formal definitions of
the evolutionary symmetry and applications of this notion to the classification
of the hyperbolic equations are discussed in details in the references cited
and many other sources, so we will not stop here. We only recall two facts.
First, even the existence of just one symmetry of high order with respect to
derivatives is a very strong condition which seems to be equivalent to the
existence of the whole hierarchy (no example of nontrivial equation is known
which possesses only one higher symmetry). By this reason, we consider only few
higher symmetries in the examples, omitting the proof that there are infinitely
many. Second, in all known examples, the symmetry algebra is decomposed into
two subalgebras containing derivatives with respect to $x$ or $y$ only. Each
evolutionary symmetry is itself an integrable equation. This property is
similar for all classes of equations under consideration, however one should
bear in mind that $y$-symmetries for equations of type (\ref{uxxy}) correspond
to the coupled systems with two dependent variables, for instance, $u$ and
$v=u_x$ (see examples in the remaining sections).

The notion of the consistent pair is the only thing from the above which may
seem unusual and in order to illustrate it we conclude the section with several
examples.

First, let us discuss the question about irreducibility. It is clear that a
consistent pair can be obtained by differentiating an equation (\ref{uxy}) and
disguising the result with some identical transformation, for instance, the
following equations are consistent:
\[
 u_{xxy}=D_x(h)+A(u_{xy}-h),\quad u_{xyy}=D_y(h)+A(u_{xy}-h)
\]
where $h=h(x,y,u,u_x,u_y)$ and $A(z)$ is an arbitrary function. However,
examples of such sort are reducible and therefore uninteresting. It is not
clear at once how to construct an irreducible pair, and after several attempts
one may suspect their existence. The following example shows that irreducible
pairs exist indeed.

\begin{example}\label{ex:Schwarz-KdV}
Let us consider equations
\begin{equation}\label{Schwarz-KdV:pair}
 u_{xxy}=\frac{u_{xy}u_{xx}}{u_x}+\frac{u^2_{xy}}{2u_y}+u_y,\qquad
 u_{xyy}=\frac{u_{xy}u_{yy}}{u_y}+\frac{u^2_{xy}}{2u_x}+u_x.
\end{equation}
It can be proved directly that identity (\ref{consistency}) holds (a simple
program for such kind of computations is presented in
\hyperref[s:app]{Appendix}).

In order to determine whether this pair is reducible, let us replace the
derivatives in virtue of an equation $u_{xy}=h(x,y,u,u_x,u_y)$ and see whether
the obtained equations can hold identically with respect to the dynamical
variables $u,u_x,u_{xx},u_y,u_{yy}$. Collecting coefficients at $u_{xx}$ in the
first equation and at $u_{yy}$ in the second one brings to relations
\[
 u_xh_{u_x}=h,\quad u_yh_{u_y}=h
\]
which imply that function $h$ must be of the form $h=u_xu_yH(x,y,u)$. Then, the
first equation of the system turns into
\[
 u^2_x(H^2+2H_u)+2u_xH_x=2.
\]
Obviously, this equation cannot be satisfied by any function $H(x,y,u)$ (let
alone one-parametric family), therefore pair (\ref{Schwarz-KdV:pair}) is
irreducible.

This pair is integrable as well, being compatible with the Schwarz--KdV
equation, for both characteristic directions:
\[
 u_t=u_{xxx}-\frac{3u^2_{xx}}{2u_x},\qquad
 u_\tau=u_{yyy}-\frac{3u^2_{yy}}{2u_y}.
\]
It should be remarked that the Schwarz--KdV equation serves as the evolutionary
symmetry not only for the pair (\ref{Schwarz-KdV:pair}), but also for the
second order equation
\begin{equation}\label{Schwarz-KdV:uxy}
 u_{xy}=\frac{2uu_xu_y}{u^2+1},
\end{equation}
and also for few other hyperbolic equations, see e.g.
\cite{Meshkov_Sokolov_2011}. In general, the correspondence between
(integrable) hyperbolic and evolutionary equations is not one-to-one: a given
hyperbolic equation correspond to one at most evolutionary symmetry of a given
order, but one and the same symmetry may correspond to different hyperbolic
equations which are not point equivalent.
\qed
\end{example}

One should not think as well that the compatibility condition
(\ref{consistency}) is related somehow with the integrability in the sense of
Definition \ref{def:integrability}. We will see in section \ref{s:B} that there
are ``approximately as much'' consistent pairs as the usual hyperbolic
equations and, apparently, the integrable cases for two classes are equally
rare. In the next example we consider a family of consistent pairs which
contains an arbitrary function and is not in general integrable. This example
illustrates also the simplest type of differential substitutions, introducing
of the potential. In general, the question about the substitutions admissible
by a given equation is difficult and its consideration is beyond the scope of
this paper. In particular, we do not know an algorithm which allows to check
the irreducibility of a consistent pair not only in the sense of Definition
\ref{def:pair}, but also modulo differential substitutions. Presumably, such an
example is provided by the pair (\ref{KN:pair}) belonging to the hierarchy of
Krichever--Novikov equation which is not related via differential substitutions
to other KdV type equations \cite{Svinolupov_Sokolov_Yamilov}.

\begin{example}\label{ex:KG}
Klein--Gordon equation
\begin{equation}\label{KG}
 q_{xy}=f'(q)
\end{equation}
admits the conservation law
\[
 D_x(f(q))=D_y(\tfrac{1}{2}q^2_x)
\]
which can be used for introducing a new variable (the potential) according to
the equations
\[
 u_x=\frac{1}{2}q^2_x,\quad u_y=f(q).
\]
Solving the second equation with respect to $q$ and substituting into the first
one brings to the equation
\[
 u_{xy}=\frac{\sqrt{2u_x}}{a'(u_y)},\quad a(f(q))=q.
\]
The potential can be introduced also in another way, according to the relations
\[
 u=q_x,\quad u_y=f'(q)
\]
which bring to the equation
\[
 u_{xy}=u/a'(u_y),\quad a(f'(q))=q.
\]
Finally, both substitutions can be mixed by adding the {\em trivial}
conservation law to the above one:
\[
 D_x(f(q)+kf'(q))=D_y(\tfrac{1}{2}q^2_x+kq_{xx}).
\]
This gives rise to the substitution
\[
 u_x=\frac{1}{2}q^2_x+kq_{xx},\quad u_y=f(q)+kf'(q)
\]
and eliminating of $q$ (as before, the latter equation is assumed to be
solvable with respect to $q$) brings to the following third order equation:
\[
 u_{xxy}=\frac{1}{ka'(u_y)}
  \bigl(u_x-(ka''(u_y)+\tfrac{1}{2}a'(u_y)^2)u^2_{xy}\bigr),\quad
 a(f(q)+kf'(q))=q.
\]
However, in this case the conservation law is not exactly equivalent to the
original equation and substituting $q$ intermediately into (\ref{KG}) provides
one more third order equation
\[
 u_{xyy}=-\frac{a''(u_y)}{a'(u_y)}u_{xy}u_{yy}+\frac{f'(a(u_y))}{a'(u_y)}.
\]
The consistency of the obtained hyperbolic pair follows from its construction
and a check along the lines of the previous example shows that it is
irreducible. What about integrability property, one can prove that it occurs
exactly in the cases when original equation (\ref{KG}) is integrable, that is,
if the function $f$ is equal to $e^q$, $e^q+e^{-q}$ or $e^q+e^{-2q}$ (up to
linear changes of $q,x,y$) corresponding to the Liouville, the sine--Gordon or
the Tzitzeica equations \cite{Zhiber_Shabat_1979}.
\qed
\end{example}

The concluding example demonstrates a further extension of the classes of
equations under consideration.

\begin{example}
The system
\begin{gather*}
 4\det\begin{pmatrix}
  u_{yy} & u_{xyy} & u_{xxyy} \\
  u_y    & u_{xy}  & u_{xxy}  \\
  u      & u_x     & u_{xx}
 \end{pmatrix}=u^3,\\
 3(u_{xy}u_{xx}-u_xu_{xxy})=u_yu_{xxx}-uu_{xxxy},\\
 3(u_{xy}u_{yy}-u_yu_{xyy})=u_xu_{yyy}-uu_{xyyy}
\end{gather*}
defines a {\em consistent triple} of fourth order hyperbolic equations, that
is, the cross-derivatives are equal identically,
\[
 D_x(u_{xxyy})=D_y(u_{xxxy}),\qquad D_y(u_{xxyy})=D_x(u_{xyyy})
\]
in virtue of the system itself. Comparing with the consistent pair
(\ref{pair}), the set of dynamical variables for such a triple contains
additionally the derivatives $u_{xxy}$ and $u_{xyy}$. It can be proved that the
above system is irreducible, that is, it cannot be obtained from some
consistent pair by differentiating. However, it is related via the substitution
$v=-2(\log u)_{xy}$ to the Tzitzeica equation in algebraic form
\[
 vv_{xy}-v_xv_y=v^3-1.
\]
More precisely, this equation gives rise to the first, trilinear equation of
the system (see e.g. \cite{Bobenko_Schief_1999}), while two bilinear ones are
consequences of the conservation laws
\[
 \Bigl(\frac{v_{xx}}{v}\Bigr)_y=3v_x,\qquad
 \Bigl(\frac{v_{yy}}{v}\Bigr)_x=3v_y.
\]
Indeed, the latter relations can be integrated after the substitution:
\[
 \frac{v_{xx}}{v}=-6(\log u)_{xx}+a(x),\qquad
 \frac{v_{yy}}{v}=-6(\log u)_{yy}+b(y),
\]
moreover one can assume $a=b=0$ without loss of generality, since the function
$u$ is defined by substitution up to arbitrary factors depending on $x$ and on
$y$. Now, replacing $v$ in the left hand sides yields two last equations of the
system.
\qed
\end{example}

\section{Potential Korteweg--de Vries equation}\label{s:pot-KdV}

It is known \cite{Meshkov_Sokolov_2011} that the pot-KdV equation
\begin{equation}\label{pot-KdV}
 u_t=u_{xxx}-3u^2_x
\end{equation}
does not admit compatible second order equations (\ref{uxy}). However, it is
compatible with the following third order equation:
\begin{equation}\label{pot-KdV:uxxy}
 u_{xxy}=\frac{u^2_{xy}-c}{2u_y}+2u_xu_y.
\end{equation}
One can prove by straightforward computation that this is the general form of
equation (\ref{uxxy}) compatible with the pot-KdV, up to the transformation
$u\to u+\alpha x+\beta y$. The parameter $c$ can be scaled either into $0$ or
$1$ by scaling $y$ and we will see that properties of the equation in two cases
are quite different (in regard of the real solutions, one should distinguish
also $c=1$ and $c=-1$, but this is not important in what follows). Equation
(\ref{pot-KdV:uxxy}) is well known, although in the different variables: this
is the potential form of the associated Camassa--Holm equation
\cite{Schiff_1998,Hone_1999}.

The full algebra of evolutionary symmetries for (\ref{pot-KdV:uxxy}) joins two
hierarchies, as in the case of equations (\ref{uxy}). One of them, pot-KdV
hierarchy, contains equations with derivatives $u_x,u_{xx},u_{xxx},\dots$ in
right hand side only, while equations belonging to the other hierarchy contain
beside $u_y,u_{yy},\dots$ also the mixed derivatives
$u_x,u_{xy},u_{xyy},\dots$. The first two members of this hierarchy are the
following, in the generic case $c\ne0$:
\begin{align}
\label{pot-KdV:T2}
 u_{\tau_2}&= u_{xy}u_{yy}-u_yu_{xyy}+u^3_y,\\
\label{pot-KdV:T3}
 u_{\tau_3}&= u_{yyy}-\frac{3u^2_{yy}}{2u_y}
    +\frac{3}{2cu_y}(u_{xy}u_{yy}-u_yu_{xyy}+u^3_y)^2.
\end{align}
Differentiating these equations with respect to $x$ and replacing $u_{xxy}$,
$u_{xxyy}$ in virtue of (\ref{pot-KdV:uxxy}) gives rise to the coupled
evolutionary systems with respect to $u$ and $u_x$. The commutativity of the
corresponding flows holds without taking equations (\ref{pot-KdV:uxxy}) into
account. These systems looks rather awkward, in particular, in the case
(\ref{pot-KdV:T2}) the matrix at the leading derivatives $u_{yy},u_{xyy}$ is
not constant and not diagonal. However, a differential substitution
$(u,u_x)\to(u,v)$ exists,
\[
  v=\frac{k-u_{xy}}{ku_y}+\frac{u}{k},\quad k^2=c,
\]
which brings the equations to more compact form:
\begin{gather*}
 k^{-1}u_{\tau_2}= u_{yy}+u^2_yv_y,\qquad
 k^{-1}v_{\tau_2}= -v_{yy}+u_yv^2_y,\\
 u_{\tau_3}= u_{yyy}+3u_yv_yu_{yy}+\frac{3}{2}u^3_yv^2_y,\qquad
 v_{\tau_3}= v_{yyy}-3u_yv_yv_{yy}+\frac{3}{2}u^2_yv^3_y.
\end{gather*}
This is the potential form of the Kaup--Newell system, or the derivative
nonlinear Schr\"odinger equation \cite{Kaup_Newell_1978}. Up to our knowledge,
its relation to equation (\ref{pot-KdV:uxxy}) and, therefore, to pot-KdV
equation (\ref{pot-KdV}) was not remarked before.

Now, let us consider the case $c=0$. First of all, notice that equation
(\ref{pot-KdV:uxxy}) acquires in this (and only this) case the first integral
\[
 D_x\left(u_{xyy}-\frac{u_{yy}u_{xy}}{u_y}-u^2_y\right)=0
\]
which we rewrite in the form
\begin{equation}\label{pot-KdV:uxyy}
 u_{xyy}=\frac{u_{yy}u_{xy}}{u_y}+u^2_y+\gamma.
\end{equation}
The integration constant $\gamma$ will be assumed independent on $y$, without
loss of generality: at $c=0$, the original equation (\ref{pot-KdV:uxxy})
becomes invariant with respect to the changes $y\to\varphi(y)$ and it is
possible to set $\gamma=\mathop{\rm const}$ by use of an appropriate
transformation.

Equations (\ref{pot-KdV:uxxy}$|_{c=0}$), (\ref{pot-KdV:uxyy}) constitute a
consistent pair. It is irreducible in the sense of Definition \ref{def:pair};
this can be easily proved by a direct check along the lines of Example
\ref{ex:Schwarz-KdV}. Nevertheless, this pair is very simply related to a
hyperbolic equation, since the substitution $u_y=e^q$ lowers the order of
equation (\ref{pot-KdV:uxyy}) and brings it to the $\sinh$-Gordon equation (or
to the Liouville equation at the special value $\gamma=0$ of the first
integral)
\[
 q_{xy}=e^q+\gamma e^{-q}.
\]
This is a particular case of the substitution from Example \ref{ex:KG}. It is
well known that an evolutionary symmetry for this equation is the pot-mKdV
equation
\begin{equation}\label{pot-mKdV}
 q_\tau=q_{yyy}-\frac{1}{2}q^3_y.
\end{equation}
Returning to the variable $u$, we obtain from here the symmetry for equation
(\ref{pot-KdV:uxyy}), namely, the Schwarz--KdV equation
\begin{equation}\label{Schwarz-KdV}
 u_\tau=u_{yyy}-\frac{3u^2_{yy}}{2u_y}
\end{equation}
which we have already meet in Example \ref{ex:Schwarz-KdV}. It is directly
proved that this symmetry is compatible, in virtue of (\ref{pot-KdV:uxyy}),
also with equations (\ref{pot-KdV:uxxy}$|_{c=0}$) and (\ref{pot-KdV}).

Thus, the existence of compatible equation (\ref{pot-KdV:uxyy}) in the case
$c=0$ brings to a conversion of the hierarchy of $y$-symmetries: it simplifies
and becomes one-component. The first symmetry (\ref{pot-KdV:T3}) turns, in
virtue of (\ref{pot-KdV:uxyy}), just into the classical symmetry
$u_{\tau_2}=-\gamma u_y$. Equation (\ref{pot-KdV:T3}) (as well as the next
equations of the hierarchy) contains the parameter $c$ in the denominator, but
this does not lead to the loss of this symmetry, since the division by $c$ can
be compensated by scaling $\gamma$ in the numerator. As the result, the
fractional term in equation (\ref{pot-KdV:T3}) also becomes proportional to
$u_y$ and can be neglected. Of course, this is just a heuristic argument, since
actually we cannot make use of the relation (\ref{pot-KdV:uxyy}) until the
parameter $c$ turns into $0$. However, as we have already said, the direct
check shows that equation (\ref{Schwarz-KdV}) defines the $y$-symmetry indeed.

\section{Potential Kaup equation}\label{s:pot-Kaup}

The potential Kaup equation
\begin{equation}\label{pot-Kaup}
 u_t=u_{xxxxx}+5u_xu_{xxx}+\frac{15}{4}u^2_{xx}+\frac{5}{3}u^3_x
\end{equation}
is compatible with the hyperbolic third order equation
\begin{equation}\label{pot-Kaup:uxxy}
 u_{xxy}=\frac{3u^2_{xy}}{4u_y}-u_xu_y-c.
\end{equation}
This is the general form of such equation, up to the change $u\to u+\beta y$,
and no compatible equation of second order exists. The parameter $c$ can be
scaled into $0$ or $1$ (the sign is not important, in contrast to the KdV
equation). Equation (\ref{pot-Kaup:uxxy}) is related to the Degasperis--Procesi
equation \cite{Degasperis_Procesi, Degasperis_Hone_Holm}.

In order to write down the $y$-symmetries denote
\[
  S=\frac{u_{yyy}}{u_y}-\frac{3u^2_{yy}}{2u^2_y},\qquad
  P=u^{-3/2}_y\Bigl(u_yu_{xyy}-u_{xy}u_{yy}+\frac{2}{3}u^3_y\Bigr),
\]
then first two higher symmetries take the form, at $c\ne0$:
\begin{align}
\label{pot-Kaup:T3}
 u_{\tau_3}&= u_y\Bigl(S+\frac{3}{4c}P^2\Bigr),\\
\label{pot-Kaup:T5}
 u_{\tau_5}&= u_y\Bigl(S_{yy}+\frac{3}{2}S^2
    +\frac{5}{4c}(2SP^2+P^2_y+2PP_{yy})+\frac{15}{16c^2}P^4\Bigr).
\end{align}
There equations can be written as two-component evolutionary systems with
respect to $u$ and $u_x$, but their form is rather bulky.

The value $c=0$ is distinguished, since the first integral appears in this
case: $D_x(P)=0$. This yields equation
\begin{equation}\label{pot-Kaup:uxyy}
 u_{xyy}=\frac{u_{yy}u_{xy}}{u_y}-\frac{2}{3}u^2_y+\gamma u^{1/2}_y
\end{equation}
which constitutes together with (\ref{pot-Kaup:uxxy}) a consistent pair. One
can assume without loss of generality, due to the changes $y\to\varphi(y)$,
that $\gamma$ does not depend on $y$.

At $c=0$, the hierarchy of $y$-symmetries becomes one-component, but, in
contrast to the previous section, its structure depends on the value of the
first integral. At the special value $\gamma=0$, the flow (\ref{pot-Kaup:T3})
survives and turns into the Schwarz--KdV equation (\ref{Schwarz-KdV}). If
$\gamma\ne0$ then the third order symmetry does not exist, while the flow
(\ref{pot-Kaup:T3}) turns into equation
\[
 u_{\tau_5}=u_{yyyyy}-5\frac{u_{yy}u_{yyyy}}{u_y}-\frac{15u^2_{yyy}}{4u_y}
  +\frac{65u^2_{yy}u_{yyy}}{4u^2_y}-\frac{135u^4_{yy}}{16u^3_y}
\]
which is compatible, in virtue of (\ref{pot-Kaup:uxyy}), with equations
(\ref{pot-Kaup:uxxy}$|_{c=0}$) and (\ref{pot-Kaup}). The dependence of the
answer on $\gamma$ becomes clear after the the substitution $u_y=e^q$ which
turns equation (\ref{pot-Kaup:uxyy}) into the Tzitzeica equation
\[
 q_{xy}=-\frac{2}{3}e^q+\gamma e^{-q/2}.
\]
If $\gamma=0$ then we obtain the Liouville equation again, with the symmetry
(\ref{pot-mKdV}).

\section{Passage to the B\"acklund variables}\label{s:B}

Till now, we used the standard dynamical variables, that is the set of
derivatives with respect to $x,y$ which cannot be eliminated in virtue of the
equation (see fig. \ref{fig:dynvar}). At a first glance, this set is the only
reasonable one. However, it turns out to be unfit in more complicated examples
like the Krichever--Novikov equation from the next section, bringing to
catastrophic computations and answers. The key to their simplification is given
by the problem of description of the consistent pairs of equations
(\ref{pair}). We restrict ourselves by the case of equations which are linear
with respect to the derivatives $u_{xx},u_{yy}$. The analysis below
demonstrates that such pairs can be conveniently represented by equations with
two independent variables $u,v$. In this analysis, the pair is not assumed to
be integrable in the sense of existence of evolutionary symmetries.

So, let us consider the pair of equations
\begin{equation}\label{pair'}
 u_{xxy}=au_{xx}+b,\qquad u_{xyy}=cu_{yy}+d
\end{equation}
where $a,b,c,d$ are functions on $x,y,u,u_x,u_y,u_{xy}$. The consistency
condition is
\begin{equation}\label{consistency'}
 D_y(au_{xx}+b)=D_x(cu_{yy}+d)
\end{equation}
and after replacing $u_{xxy}$ and $u_{xyy}$ in virtue of (\ref{pair'}) this
equation must hold identically on $x,y,u$ and derivatives of $u$.

\begin{theorem}
If system (\ref{pair'}) is consistent and irreducible then it is of the form
\begin{equation}\label{pair''}
\begin{aligned}[b]
 u_{xxy}&=\frac{1}{h_{u_{xy}}}\bigl(F(x,y,u,u_x,h)
   -h_x-h_uu_x-h_{u_x}u_{xx}-h_{u_y}u_{xy}\bigr),\\
 u_{xyy}&=\frac{1}{h_{u_{xy}}}\bigl(G(x,y,u,u_y,h)
   -h_y-h_uu_y-h_{u_x}u_{xy}-h_{u_y}u_{yy}\bigr)
\end{aligned}
\end{equation}
where function $h=h(x,y,u,u_x,u_y,u_{xy})$ is implicitly defined by equation
\begin{equation}\label{fgh}
 F_hG-FG_h+u_{xy}(F_{u_x}-G_{u_y})+u_yF_u-u_xG_u+F_y-G_x=0.
\end{equation}
\end{theorem}
\begin{proof}
Collecting terms with $u_{xx}u_{yy}$ in equation (\ref{consistency'}) yields
the relation
\[
 a_{u_y}-c_{u_x}=ac_{u_{xy}}-ca_{u_{xy}}
\]
which serves as the compatibility condition for the system of equations
\[
 h_{u_x}=-ah_{u_{xy}},\quad h_{u_y}=-ch_{u_{xy}}
\]
with respect to an unknown function $h$. Therefore, if the pair is consistent
then a function $h(x,y,u,u_x,u_y,u_{xy})$ exists such that it is represented as
follows:
\[
 u_{xxy}=f_0-\frac{h_{u_x}}{h_{u_{xy}}}u_{xx},\quad
 u_{xyy}=g_0-\frac{h_{u_y}}{h_{u_{xy}}}u_{yy}.
\]
It is convenient to redefine $f_0,g_0$ and to rewrite this in the form
\[
 D_x(h)=\tilde f(x,y,u,u_x,u_y,u_{xy}),\qquad
 D_y(h)=\tilde g(x,y,u,u_x,u_y,u_{xy}).
\]
Now, collecting the terms with $u_{xx}$ and $u_{yy}$ in equation
(\ref{consistency'}) yields
\[
 \tilde f_{u_y}h_{u_{xy}}=\tilde f_{u_{xy}}h_{u_y},\qquad
 \tilde g_{u_x}h_{u_{xy}}=\tilde g_{u_{xy}}h_{u_x}.
\]
Solving these equation proves that functions $\tilde f,\tilde g$ are of the
form
\[
 \tilde f=F(x,y,u,u_x,h),\qquad \tilde g=G(x,y,u,u_y,h),
\]
that is the pair (\ref{pair'}) is of the form (\ref{pair''}). Moreover, the
consistency condition now takes the form (\ref{fgh}).

Now let us prove that if the pair is irreducible then equation (\ref{fgh}) must
be (locally) solvable with respect to $h$. Assume that this is not the case,
that is, the functions $F$ and $G$ are such that equation (\ref{fgh}) holds
identically on $h$. Then the coefficient at $u_{xy}$ must vanish:
$F_{u_x}-G_{u_y}=0$, because $F$ and $G$ do not depend on $u_{xy}$
intermediately. This implies that $F$ is linear with respect to $u_x$, $G$ is
linear with respect to $u_y$ and further analysis of equation (\ref{fgh}) shows
easily that these functions are of the form
\[
 F=-\frac{s_uu_x+s_x}{s_h},\quad G=-\frac{s_uu_y+s_y}{s_h},\quad s=s(x,y,u,h).
\]
However, in such a case the pair (\ref{pair''}) is reducible because its
equations are obtained by differentiating of the one-parametric family of
second order equations $s(x,y,u,h(x,y,u,u_x,u_y,u_{xy}))=\alpha$.
\end{proof}

Equation (\ref{fgh}) can be effectively solved with respect to $h$ only for
very special functions $F,G$, for example, linear or quadratic with respect to
$h$. As a particular example, if we chose
\[
 F=\frac{(u-h)^2}{u_x},\quad G=\frac{(u-h)^2}{u_y}
\]
then equation (\ref{fgh}) is linear with respect to $h$ and we find
\[
 h=u-2\frac{u_xu_y}{u_{xy}}.
\]
Now, substitution into equations (\ref{pair''}) brings to pair
(\ref{Schwarz-KdV:pair}) from Example \ref{ex:Schwarz-KdV}.

If functions $F,G$ are of more general form then representation (\ref{pair''})
is practically useless. However, it brings to the idea that the ``proper''
dynamical variable is not the mixed derivative $u_{xy}$, but the function $h$
itself. Indeed, if we introduce the new variable $v=h(x,y,u,u_x,u_y,u_{xy})$
then $u_{xy}$ is defined explicitly as the inverse function of $v$, since it
enters into (\ref{fgh}) linearly. It is easy to see that the result of this
transformation is that equations (\ref{pair''}) are written in the form
\begin{equation}\label{vxvy}
 v_x=F(x,y,u,u_x,v),\qquad v_y=G(x,y,u,u_y,v),
\end{equation}
while equation (\ref{fgh}) follows from here after eliminating of the cross
derivatives and provides an equation for $u_{xy}$:
\begin{equation}\label{H}
 u_{xy}=H(x,y,u,u_x,u_y,v).
\end{equation}
These equations give the desired representation of the consistent pair. In
order to return to the standard set of dynamical variables, one has to
differentiate (\ref{H}) with respect to $x$ and $y$ and to eliminate $v$ from
the obtained equations, using equations (\ref{vxvy}), (\ref{H}) again.

Notice that the roles of the variables $u$ and $v$ in system (\ref{vxvy}) are
completely equal (for generic functions $F,G$) and eliminating of $u$ instead
of $v$ yields a consistent pair of hyperbolic equations with respect to the
variable $v$. Exactly the same idea is used in the definition of B\"acklund
transformation (see e.g. \cite{Ablowitz_Segur_1981}), this is why we call $u,v$
the B\"acklund variables.

\begin{definition}\label{def:Bpair}
We say that equations (\ref{vxvy}) define the representation of the pair
(\ref{pair'}) in the B\"acklund variables if equations (\ref{pair'}) follows
from (\ref{vxvy}) as a result of eliminating of the variable $v$.
\end{definition}

Let us stress that pair (\ref{pair'}) which admits a representation
(\ref{vxvy}) is automatically consistent.

In the above proof, we have seen that equation (\ref{fgh}) may lose the
dependence on $h$ or $u_{xy}$ under some special choices of $F,G$. Such
functions are unfit for the definition of pair (\ref{pair''}). However, one can
waive these restrictions and consider systems (\ref{vxvy}) with arbitrary $F$
and $G$ as a basic object. From this point of view, degenerations of different
types are admissible. In particular, equations (\ref{vxvy}) may define indeed a
B\"acklund transformation between hyperbolic equations of second order. This
correspond to the situation when equation (\ref{H}) does not contain $v$ and
analogous equation for $v_{xy}$ does not contain $u$. Thus, the class of
systems (\ref{vxvy}) is rather general and significant. Functions $F,G$ depend
on 5 arguments, that is the functional dimension of this class is the same as
for the class of second order equations (\ref{uxy}).

It should be remarked that representation in the B\"acklund variables is not
unique, in contrast to the representation in the standard dynamical variables.
Indeed, variable $v$ can be replaced by any variable of the form
\begin{equation}\label{newv}
 \tilde v=\varphi(x,y,u,v),\quad \varphi_v\ne0
\end{equation}
without changing the general form of equations (\ref{vxvy}). Certainly,
eliminating $\tilde v$ results in the same equations for $u$ as before. This
arbitrariness should be taken into account when bringing a given system
(\ref{vxvy}) to a simpler form.

A remarkable feature is that B\"acklund variables may turn convenient even in
consideration of a single third order equation rather than a pair
(\ref{pair'}). In such a case, one equation of system (\ref{vxvy}) is replaced
with an equation of the form (\ref{H}). As an example, let us rewrite in the
B\"acklund variables some equations corresponding to the pot-KdV equation
(\ref{pot-KdV}). The consistent pair (\ref{pot-KdV:uxxy}$|_{c=0}$),
(\ref{pot-KdV:uxyy}) can be represented as the system
\[
 v_x+u_x=\frac{1}{2}(u-v)^2,\qquad v_yu_y=-\gamma.
\]
Notice that its first equation defines the $x$-part of the B\"acklund
transformation (with zero spectral parameter) for equation (\ref{pot-KdV}). At
$c\ne0$, the third order equation (\ref{pot-KdV:uxxy}) is equivalent to the
system
\[
 v_x+u_x=\frac{1}{2}(u-v)^2,\qquad u_{xy}=u_y(u-v)+k,\qquad k^2=c.
\]
This result is in a close relation with the representation of the associated
Camassa--Holm equation by compatible differential-difference equations, the
dressing chain and a Volterra-type lattice
\cite{Shabat_Yamilov_1991,Adler_Shabat_2006}.

Analogously, the consistent pair (\ref{pot-Kaup:uxxy}$|_{c=0}$),
(\ref{pot-Kaup:uxyy}) corresponding to the Kaup equation (\ref{pot-Kaup}) can
be represented as the system
\[
 v_x=\frac{1}{4}(u+v)^2,\qquad v_y=-\frac{1}{3}u_y-\gamma u^{-1/2}_y,
\]
while single equation (\ref{pot-Kaup:uxxy}) is equivalent to the system
\[
 v_x=\frac{1}{4}(u+v)^2,\qquad u_{xy}=-u_y(u+v)+ku^{1/2}_y,\qquad k^2=4c.
\]
A more complicated example of using the B\"acklund variables is presented in
the next section.

\section{Krichever--Novikov equation}\label{s:KN}

It is known \cite{Meshkov_Sokolov_2011} that the Krichever--Novikov equation
\begin{equation}\label{KN}
 u_t=u_{xxx}-\frac{3(u^2_{xx}-r(u))}{2u_x},\quad r^{(5)}(u)=0
\end{equation}
does not admit a consistent second order hyperbolic equation (\ref{uxy}). In
this example, the search of compatible third order equation using the standard
dynamical variables runs into inextricable computational difficulties. The
computations in the B\"acklund variables bring to the following answer.

\begin{theorem}
Let $h=h(u,v)$ be a biquadratic polynomial and $r(u)$, $s(v)$ be its
discriminants with respect to $v$, $u$, respectively:
\[
 h_{uuu}=h_{vvv}=0,\quad r(u)=h^2_v-2hh_{vv},\quad s(v)=h^2_u-2hh_{uu}.
\]
Then: 1) equation (\ref{KN}) defines the evolutionary symmetry for the
hyperbolic system (with arbitrary parameter $c$)
\begin{gather}
\label{KN:uxy}
 u_{xy}=\frac{u_x}{h}(h_uu_y+c(hh_{uv}-h_uh_v)),\\
\label{KN:uxvx}
 u_xv_x=h(u,v);
\end{gather}
2) the following equation holds in virtue of equations (\ref{KN}),
(\ref{KN:uxy}), (\ref{KN:uxvx}):
\begin{equation}\label{KN:vt}
 v_t=v_{xxx}-\frac{3(v^2_{xx}-s(v))}{2v_x},\quad s^{(5)}(v)=0;
\end{equation}
3) the variables $u$ and $v$ are on equal footing: equation (\ref{KN:uxy}) can
be replaced with
\begin{equation}\label{KN:vxy}
 v_{xy}=\frac{v_x}{h}(h_vv_y-c(hh_{uv}-h_uh_v)).
\end{equation}
\end{theorem}
\begin{proof}
In order to prove statement 1) one should verify that differentiating system
(\ref{KN:uxy}), (\ref{KN:uxvx}) with respect to $t$ in virtue of (\ref{KN})
gives rise to an identity modulo the system. First, differentiating of equation
(\ref{KN:uxy}) yields the equality
\begin{multline*}
 \biggl(D_xD_y -\frac{h_uu_y+c\hat h}{h}D_x -\frac{h_u}{h}u_xD_y
   -\Bigl(\frac{h_u}{h}\Bigr)_uu_xu_y \\
   -c\Bigl(\frac{\hat h}{h}\Bigr)_uu_x\biggr)
 \left(u_{xxx}-\frac{3(u^2_{xx}-r(u))}{2u_x}\right)
  = u_x\biggl(\Bigl(\frac{h_u}{h}\Bigr)_vu_y
     +c\Bigl(\frac{\hat h}{h}\Bigr)_v\biggr)v_t
\end{multline*}
where we denote $\hat h=hh_{uv}-h_uh_v$, and an explicit expression for $v_t$
is obtained from here. It is simplified after replacing the derivatives of $u$
in virtue of (\ref{KN:uxy}), (\ref{KN:uxvx}), and equation (\ref{KN:vt})
appears as a result of straightforward, although rather tedious computations (a
remarkable circumstance here is that the left hand side of the equation is
divisible by the expression in the brackets from the right hand side).

Thus, statement 2) is proved and in order to complete the proof of statement
1), we have to check that differentiating equation (\ref{KN:uxvx}) with respect
to $t$ yields an identity, that is
\begin{multline*}\qquad
  (v_xD_x-h_u)\left(u_{xxx}-\frac{3(u^2_{xx}-r(u))}{2u_x}\right) \\
 +(u_xD_x-h_v)\left(v_{xxx}-\frac{3(v^2_{xx}-s(v))}{2v_x}\right)=0.\qquad
\end{multline*}
This is proved by a direct and relatively simple computation. Moreover, since
the original equation (\ref{KN}) and the obtained equation (\ref{KN:vt}) do not
contain the derivatives with respect to $y$, hence equation (\ref{KN:uxy}) is
not actually needed in this computation and only relation (\ref{KN:uxvx}) is
used. This is exactly equivalent to the known result \cite{Adler_1998} that
this relation defines the $x$-part of the B\"acklund transformation between
equations (\ref{KN}) and (\ref{KN:vt}).

Statement 3) is very simple:
\begin{align*}
 v_{xy}&=\left(\frac{h}{u_x}\right)_y=\frac{h_uu_y+h_vv_y}{u_x}
    -\frac{hu_{xy}}{u^2_x}\\
 &= \frac{h_uu_y}{u_x}+\frac{h_vv_yv_x}{h}
    -\frac{1}{u_x}(h_uu_y+c(hh_{uv}-h_uh_v))\\
 &= \frac{v_x}{h}(h_vv_y-c(hh_{uv}-h_uh_v)).
\end{align*}
It is clear that, vice versa, (\ref{KN:uxy}) follows from (\ref{KN:vxy}).
\end{proof}

Notice that equations (\ref{KN})--(\ref{KN:vxy}) keep the form invariant under
the M\"obius transformations of $u$ and $v$. The orbits of the group of
transformations depend, in particular, on the multiplicity of the zeroes of the
polynomials $r,s$ (the proper Krichever--Novikov equation corresponds to the
generic case of simple zeroes) and the detailed classification of the orbits is
contained in \cite{Adler_Bobenko_Suris_2009}.

Thus, hyperbolic system (\ref{KN:uxy}), (\ref{KN:uxvx}) defines a certain
extension of the B\"acklund transformation for the Krichever--Novikov equation
(one can call it $y$-part). An important feature is that in the case of usual
B\"acklund transformation the variables $u$ and $v$ cannot be explicitly
expressed one from another, while adding the new independent variable $y$ makes
the transformation explicit. Indeed, variable $v$ is expressed through
$u,u_x,u_y$ and $u_{xy}$ as a solution of equation (\ref{KN:uxy}), and the
inverse transformation is obtained from (\ref{KN:vxy}). In particular, this is
why equation (\ref{KN:vt}) is uniquely derived (\ref{KN})--(\ref{KN:uxvx}) and
should not be postulated in advance.

The polynomial $\hat h=hh_{uv}-h_uh_v$ in the right hand side of equation
(\ref{KN:uxy}) is biquadratic as well (moreover, an algebraic identity holds
$\hat h\hat h_{uv}-\hat h_u\hat h_v=\mathop{\rm const}h$). Therefore, in order
to find $v$, one has to solve a quadratic equation with rather cumbersome
coefficients. In principle, after substituting the obtained expression into
(\ref{KN:uxvx}), one can write down a hyperbolic third order equation in the
standard dynamical variables (that is, of form (\ref{uxxy})), but it is so
bulky that the inapplicability of these variables in this example becomes
obvious. Moreover, the form of this equation depends on the particular choice
of $h$. In the most degenerate case $r=0$, $h=\frac{1}{2}(u-v)^2$ corresponding
to the Schwarz--KdV equation an essential simplifications occur and one obtains
the equation (cf with pair (\ref{Schwarz-KdV:pair}) and equation
(\ref{pot-KdV:uxyy}))
\begin{equation}\label{Schwarz-KdV:uxxy}
 u_{xxy}=\frac{u_{xy}u_{xx}}{u_x}+\frac{u^2_{xy}-c^2u^2_x}{2u_y}+u_y.
\end{equation}

It is natural to use the variables $u,v$ also for representing of
$y$-symmetries for system (\ref{KN:uxy}), (\ref{KN:uxvx}) (recall that, in the
standard dynamical set, we used $u_x$ instead of $v$). The two simplest higher
flows presented in the following statement were found by an intermediate
computation. These flows are defined just by one equation for $u_\tau$, since
equation for $v_\tau$ is derived automatically. However, for the sake of
completeness, we write down both equations of the coupled system.

\begin{statement}
If $c\ne0$ then system (\ref{KN:uxy}), (\ref{KN:uxvx}) admits the following
symmetries:
\begin{equation}\label{KN:uT2}
\begin{aligned}[b]
  u_{\tau_2}&=u_{yy}-\frac{1}{ch}(u^2_y-c^2r(u))(v_y+ch_u)-\frac{c^2}{2}r'(u),\\
 -v_{\tau_3}&=v_{yy}+\frac{1}{ch}(v^2_y-c^2s(v))(u_y-ch_v)-\frac{c^2}{2}s'(v);
\end{aligned}
\end{equation}
\begin{equation}\label{KN:uT3}
\begin{aligned}[b]
 u_{\tau_3}&=u_{yyy}-\frac{3u_yu_{yy}}{ch}(v_y+ch_u)
      +\frac{3u_y}{2c^2h^2}(u^2_y-c^2r(u))(v_y+ch_u)^2\\
 &\qquad +\frac{3c}{2h}u_yv_yr'(u)+\frac{3c^2}{2h^2}(h^2_ur(u)-\hat h^2)u_y,\\
 v_{\tau_3}&=v_{yyy}+\frac{3v_yv_{yy}}{ch}(u_y-ch_v)
      +\frac{3v_y}{2c^2h^2}(v^2_y-c^2s(v))(u_y-ch_v)^2\\
 &\qquad -\frac{3c}{2h}u_yv_ys'(v)+\frac{3c^2}{2h^2}(h^2_vs(v)-\hat h^2)v_y.
\end{aligned}
\end{equation}
\end{statement}

These are equations from the modified Landau--Lifshitz hierarchy, written under
the stereographic projection \cite{Shabat_Yamilov_1991}.

Like in examples from sections \ref{s:pot-KdV}, \ref{s:pot-Kaup}, the parameter
$c$ can be scaled to $0$ or $1$ and the properties of equations are different
in the two cases. It is easy to prove that if $c=0$ then system (\ref{KN:uxy}),
(\ref{KN:uxvx}) admits the first integral
\[
  D_x\left(\frac{u_yv_y}{h}\right)=0.
\]
Up to the changes $y\to\varphi(y)$, the value of the integration constant can
be chosen equal to $1$, and this brings to the following hyperbolic pair
written in the B\"acklund variables (see Definition \ref{def:Bpair}).

\begin{statement}
The system
\begin{equation}\label{KN:pair}
 u_xv_x=h(u,v),\quad u_yv_y=h(u,v)
\end{equation}
admits the Krichever--Novikov equation both as $x$- and $y$-symmetry, that is,
it is compatible with equation (\ref{KN}) and equation
\[
 u_\tau=u_{yyy}-\frac{3(u^2_{yy}-r(u))}{2u_y}.
\]
\end{statement}

The consistency of (\ref{KN}) and (\ref{KN:pair}) follows from the construction
and the formula for $y$-symmetry is obvious since the independent variables $x$
and $y$ are now on the equal footing. Of course, the statement can also be
proven directly. However, the passage to the limit from system (\ref{KN:uT3})
is more complicated in this example comparing to the cases of the KdV and the
Kaup equations. Indeed, the right hand side of (\ref{KN:uT3}) even does not
contain the terms with $u^2_{yy}$. It is clear that this paradox is explained
by the fact that the second equation (\ref{KN:pair}) is not valid until $c$
turns into $0$ and it should be replaced with a certain formal power series
with respect to $c$, but we will not dive into this analysis. The flow
(\ref{KN:uT2}) turns into the classical symmetry $u_\tau=u^2_yv_y/h=u_y$ after
multiplying the right hand side by $c$ and setting $c=0$.

It may be not clear in the above exposition, wherefrom the hyperbolic system
(\ref{KN:uxy}), (\ref{KN:uxvx}) appears, especially if one does not wish to
employ the fact that (\ref{KN:uxvx}) defines the $x$-part of the B\"acklund
transformation for (\ref{KN}). Actually, any guess is not needed. The search of
a pair (\ref{vxvy}) which is compatible with a given evolutionary equation is
just a matter of computation and if we start from the Krichever--Novikov
equation then it quickly leads us to system (\ref{KN:pair}). The only one step
which is not algorithmic here is the choice of a convenient gauge (\ref{newv}),
but in this example it is quite obvious from the symmetry arguments. When the
B\"acklund variables are chosen, we may forget about the second equation
(\ref{vxvy}) and search for more general equation of the form (\ref{H}) which
is compatible with the $t$-dynamics. This leads to equation (\ref{KN:uxy}) with
the additional parameter $c$.

\section{Conclusion}

In this paper we have considered several examples of third order hyperbolic
equations possessing higher evolutionary symmetries. These examples demonstrate
that equations of such type may acquire a first integral under a parametric
degeneration which is interpreted as a complementary hyperbolic equation
consistent with the original one. Such consistent pairs of equations are
interesting objects as themselves. Their study leads us to the notion of the
B\"acklund variables which provide convenient dynamical variables for the
equation. The example of the Krichever--Novikov equation suggests that
introducing such variables is a natural and reasonable step if we wish to
obtain a complete description of integrable hyperbolic equations. However, this
classification problem seems rather difficult and it hardly can be solved in
the near future. Therefore, the actual problems are the search of new examples,
further study of the associated structures and construction of explicit
solutions.

The discrete analogs of hyperbolic equations (\ref{uxy}) are lattice equations
of the form
\[
 h_{m,n}(u_{m,n},u_{m+1,n},u_{m,n+1},u_{m+1,n+1})=0
\]
(see e.g. \cite{Adler_Bobenko_Suris_2009}). Recall, that their theory is also
closely related with the B\"acklund transformations: first, such equations
define the nonlinear superposition principle for equations (\ref{uxy}) and KdV
type equations, second, they define the B\"acklund transformations for
differential-difference equations of the Volterra lattice type. Like in the
continuous case, some examples require consideration of more general types of
equations. In particular, it is interesting to extend the obtained results to
equations of the form
\[
 f(u_{m,n},u_{m+1,n},u_{m+2,n},u_{m,n+1},u_{m+1,n+1},u_{m+2,n+1})=0
\]
which play the role of discrete analogs of equations (\ref{uxxy}).

\section*{Acknowledgements}
\addcontentsline{toc}{section}{Acknowledgements}

This work was partially supported by grants NSh--5377.2012.2 and RFBR
10-01-00088, 11-01-12018.

\section*{Appendix}
\addcontentsline{toc}{section}{Appendix}
\label{s:app}

Here we bring a sample {\em Mathematica} \cite{Wolfram_2003} program which
allows to check the compatibility of the Schwarz--KdV equation
\[
 u_t=u_{xxx}-\frac{3u^2_{xx}}{2u_x}
\]
with hyperbolic equations of different types. The computation is performed in
the standard set of dynamical variables (an implementation of the B\"acklund
ones requires certain modifications). The partial derivative
$\partial^m_x\partial^n_y(u)$ is denoted \verb|u[m,n]|. First, the operators of
the total derivatives with respect to $x,y$ and $t$ are defined, as well as the
equation itself:
\begin{quote}
\begin{verbatim}
 vars[f_]:=Union[Cases[f,_u,{0,\[Infinity]}]]
 diff[f_]:=Apply[Plus,Map[D[f,#]dif[#]&,vars[f]]]
 dx[f_]:=D[f,x]+diff[f]/.dif[u[m_,n_]]:>u[m+1,n]
 dx[f_,n_]:=Nest[dx,f,n]
 dy[f_]:=D[f,y]+diff[f]/.dif[u[m_,n_]]:>u[m,n+1]
 dy[f_,n_]:=Nest[dy,f,n]
 dt[f_]:=diff[f]/.dif[u[m_,n_]]:>dx[dy[ut,n],m]
 ut=u[3,0]-3/2*u[2,0]^2/u[1,0];
\end{verbatim}
\end{quote}
In the next lines, the mixed derivatives are eliminated in virtue of second
order hyperbolic equation (\ref{Schwarz-KdV:uxy}) and its compatibility with
the flow $\partial_t$ is verified:
\begin{quote}
\begin{verbatim}
 Clear[u]
 u[n_,1]:=dx[uxy,n-1]/;n>0
 u[1,n_]:=dy[uxy,n-1]/;n>0
 uxy=2*u[0,0]*u[1,0]*u[0,1]/(u[0,0]^2+1);
 Together[dx[dy[ut]]-dt[uxy]]
\end{verbatim}
\end{quote}
Analogously, the following fragment handles third order equation
(\ref{Schwarz-KdV:uxxy}):
\begin{quote}
\begin{verbatim}
 Clear[u]
 u[n_,1]:=dx[uxxy,n-2]/;n>1
 u[2,n_]:=dy[uxxy,n-1]/;n>0
 uxxy=u[1,1]*u[2,0]/u[1,0]+
   (u[1,1]^2-c^2*u[1,0]^2)/(2*u[0,1])+u[0,1];
 Together[dx[dy[ut],2]-dt[uxxy]]
\end{verbatim}
\end{quote}
Finally, consistent pair (\ref{Schwarz-KdV:pair}) is checked below. The last
three lines verify consistency of the pair itself, that is identity
(\ref{consistency}), and the compatibility of each equation of the pair with
the flow $\partial_t$:
\begin{quote}
\begin{verbatim}
 Clear[u]
 u[n_,1]:=dx[uxxy,n-2]/;n>1
 u[1,n_]:=dy[uxyy,n-2]/;n>1
 uxxy=u[1,1]*u[2,0]/u[1,0]+u[1,1]^2/(2*u[0,1])+u[0,1];
 uxyy=u[1,1]*u[0,2]/u[0,1]+u[1,1]^2/(2*u[1,0])+u[1,0];
 Together[dx[uxyy]-dy[uxxy]]
 Together[dx[dy[ut],2]-dt[uxxy]]
 Together[dx[dy[ut,2]]-dt[uxyy]]
\end{verbatim}
\end{quote}

\phantomsection

\end{document}